\documentclass[sloppy,oribibl,envcountsame,envcountsect,11pt]{article}

\usepackage[utf8]{inputenc}
\usepackage{amsmath,amssymb}

\usepackage{a4wide}

\title{On the Complexity of Noncommutative Polynomial Factorization}

\author{V. Arvind\thanks{Institute of Mathematical Sciences, Chennai,
    India, \texttt{email: arvind@imsc.res.in}} \and Pushkar S
  Joglekar\thanks{Vishwakarma Institute of Technology, Pune, India,
    \texttt{email: joglekar.pushkar@gmail.com}} \and Gaurav
  Rattan\thanks{Institute of Mathematical Sciences, Chennai, India,
    \texttt{email: grattan@imsc.res.in}}}

\date{}

\newtheorem{theorem}{Theorem}[section]
\newtheorem{corollary}[theorem]{Corollary}
\newtheorem{definition}[theorem]{Definition}
\newtheorem{lemma}[theorem]{Lemma}

\newtheorem{claim}[theorem]{Claim}
\newtheorem{remark}[theorem]{Remark}

\newenvironment{proof} {\noindent{\it Proof. }} {{\qed}}
\newenvironment{proofof}[1]{\noindent{\it Proof of #1. }} {{\qed}}

\def\qed{\hspace*{\fill} $\Box$\par\medskip}

\newcommand{\F}{\mathbb{F}}

\renewcommand{\angle}[1]{\langle #1 \rangle}

\newcommand{\FX}{\F\angle{X}}

\newcommand{\fx}{\F\angle{x_1,\ldots,x_n}}

\newcommand{\pd}{\partial}

\newcommand{\NC}{\mathrm{NC}}

\newcommand{\MA}{\mathrm{MA}}

\newcommand{\coNP}{\mathrm{coNP}}

\newcommand{\NP}{\mathrm{NP}}

\newcommand{\size}{\mathrm{size}}

\newcommand{\var}{\mathrm{Var}}

\newcommand{\poly}{\mathrm{poly}}

\newcommand{\SOP}{\textsc{SOP}}

\begin{document}
\maketitle 

\begin{abstract}
In this paper we study the complexity of factorization of polynomials
in the free noncommutative ring $\F\angle{x_1,x_2,\ldots,x_n}$ of
polynomials over the field $\F$ and noncommuting variables
$x_1,x_2,\ldots,x_n$. Our main results are the following:

\begin{itemize}
\item[$\bullet$] Although $\fx$ is not a unique factorization ring, we
  note that \emph{variable-disjoint} factorization in $\fx$ has the
  uniqueness property. Furthermore, we prove that computing the
  variable-disjoint factorization is polynomial-time equivalent to
  Polynomial Identity Testing (both when the input polynomial is given
  by an arithmetic circuit or an algebraic branching program). We also
  show that variable-disjoint factorization in the black-box setting
  can be efficiently computed (where the factors computed will be also
  given by black-boxes, analogous to the work \cite{KT91} in the
  commutative setting).

\item[$\bullet$] As a consequence of the previous result we show that
  homogeneous noncommutative polynomials and multilinear
  noncommutative polynomials have unique factorizations in the usual
  sense, which can be efficiently computed.

\item[$\bullet$] Finally, we discuss a polynomial decomposition
  problem in $\fx$ which is a natural generalization of homogeneous
  polynomial factorization and prove some complexity bounds for it.
\end{itemize}
\end{abstract}

\section{Introduction}\label{one}

Let $\F$ be any field and $X=\{x_1,x_2,\ldots,x_n\}$ be a set of $n$
free noncommuting variables. Let $X^*$ denote the set of all free
words (which are monomials) over the alphabet $X$ with concatenation
of words as the monoid operation and the empty word $\epsilon$ as
identity element.

The \emph{free noncommutative ring} $\FX$ consists of all finite
$\F$-linear combinations of monomials in $X^*$, where the ring
addition $+$ is coefficient-wise addition and the ring multiplication
$*$ is the usual convolution product. More precisely, let $f,g\in\FX$
and let $f(m)\in\F$ denote the coefficient of monomial $m$ in
polynomial $f$. Then we can write $f=\sum_m f(m) m$ and $g=\sum_m g(m)
m$, and in the product polynomial $fg$ for each monomial $m$ we have

\[
fg(m)=\sum_{m_1m_2=m} f(m_1)g(m_2).
\]

The \emph{degree} of a monomial $m\in X^*$ is the length of the
monomial $m$, and the degree $\deg f$ of a polynomial $f\in\FX$ is the
degree of a largest degree monomial in $f$ with nonzero coefficient.
For polynomials $f,g\in\FX$ we clearly have $\deg (fg) = \deg f + \deg
g$.

A \emph{nontrivial factorization} of a polynomial $f\in\FX$ is an
expression of $f$ as a product $f=gh$ of polynomials $g,h\in\FX$ such
that $\deg g > 0$ and $\deg h > 0$. A polynomial $f\in\FX$ is
\emph{irreducible} if it has no nontrivial factorization and is
\emph{reducible} otherwise.  For instance, all degree $1$ polynomials
in $\FX$ are irreducible. Clearly, by repeated factorization every
polynomial in $\FX$ can be expressed as a product of
irreducibles. 

In this paper we study the algorithmic complexity of polynomial
factorization in the free ring $\FX$.

\subsection*{Polynomial Factorization Problem}

The problem of polynomial factorization in the \emph{commutative}
polynomial ring $\F[x_1,x_2,\ldots,x_n]$ is a classical well-studied
problem in algorithmic complexity culminating in Kaltofen's celebrated
efficient factorization algorithm \cite{Kalt}. Kaltofen's algorithm
builds on efficient algorithms for univariate polynomial
factorization; there are deterministic polynomial-time algorithms over
rationals and over fields of unary characteristic and randomized
polynomial-time over large characteristic fields (the textbook
\cite{vzGbook} contains a comprehensive excellent treatment of the
subject). The basic idea in Kaltofen's algorithm is essentially a
randomized reduction from multivariate factorization to univariate
factorization using Hilbert's irreducibility theorem. Thus, we can say
that Kaltofen's algorithm uses randomization in two ways: the first is
in the application of Hilbert's irreducibility theorem, and the second
is in dealing with \emph{univariate} polynomial factorization over
fields of large characteristic. In a recent paper Kopparty et al
\cite{KSS} have shown that the first of these requirements of
randomness can be eliminated, assuming an efficient algorithm as
subroutine for the problem of \emph{polynomial identity testing} for
small degree polynomials given by circuits. More precisely, it is
shown in \cite{KSS} that over finite fields of unary characteristic
(or over rationals) polynomial identity testing is deterministic
polynomial-time equivalent to multivariate polynomial factorization.

Thus, in the commutative setting it turns out that the complexity of
multivariate polynomial factorization is closely related to polynomial
identity testing (whose deterministic complexity is known to be
related to proving superpolynomial size arithmetic circuit lower bounds).


\subsection*{Noncommutative Polynomial Factorization}

The study of noncommutative arithmetic computation was initiated by
Nisan \cite{Ni91} in which he showed exponential size lower bounds for
algebraic branching programs that compute the noncommutative permanent
or the noncommutative determinant. Noncommutative polynomial identity
testing was studied in \cite{BW05,RS05}. In \cite{BW05} a randomized
polynomial time algorithm is shown for identity testing of polynomial
degree noncommutative arithmetic circuits. For algebraic branching
programs \cite{RS05} give a deterministic polynomial-time
algorithm. Proving superpolynomial size lower bounds for
noncommutative arithmetic circuits computing the noncommutative
permanent is open. Likewise, obtaining a deterministic polynomial-time
identity test for polynomial degree noncommutative circuits is open.\\

In this context, it is interesting to ask if we can relate the
complexity of noncommutative factorization to noncommutative
polynomial identity testing. However, there are various mathematical
issues that arise in the study of noncommutative polynomial
factorization.

Unlike in the commutative setting, the noncommutative polynomial ring
$\FX$ is \emph{not} a unique factorization ring. A well-known example
is the polynomial
\[
xyx+x
\]
which has two factorizations: $x(yx+1)$ and $(xy+1)x$. Both $xy+1$ and
$yx+1$ are irreducible polynomials in $\FX$. 

There is a detailed theory of factorization in noncommutative
rings~\cite{cohn,cohn-ufd}. We will mention an interesting result
on the structure of polynomial factorizations in the ring $R=\FX$.

Two elements $a, b\in R$ are \emph{similar} if there are elements $a',
b'\in R$ such that $ab'=a'b$, and (i) $a$ and $a'$ do not have common
nontrivial left factors, (ii) $b$ and $b'$ do not have common
nontrivial right factors.

Among other results, Cohn \cite{cohn-ufd} has shown the following
interesting theorem about factorizations in the ring $R=\FX$.

\begin{theorem}[Cohn's theorem]
For a polynomial $a\in\FX$ let
\[
a  =  a_1a_2\dots a_r \textrm{ and } a  =  b_1b_2\dots b_s
\]
be any two factorizations of $a$ into irreducible polynomials in
$\FX$.  Then $r=s$, and there is a permutation $\pi$ of the indices
$\{1,2,\dots,r\}$ such that $a_i$ and $b_{\pi(i)}$ are similar
  polynomials for $1\le i\le r$.
\end{theorem}

For instance, consider the two factorizations of $xyx+x$ above. We
note that polynomials $xy+1$ and $yx+1$ are similar. It is easy to
construct examples of degree $d$ polynomials in $\FX$ that have
$2^{\Omega(d)}$ distinct factorizations. Cohn \cite{cohn} discusses a
number of interesting properties of factorizations in $\FX$. But it is
not clear how to algorithmically exploit these to obtain an efficient
algorithm in the general case.\\

\noindent\textbf{Our Results}\\

In this paper, we study some \emph{restricted} cases of polynomial
factorization in the ring $\FX$ and prove the following results.

\begin{itemize}

\item[$\bullet$] We consider \emph{variable-disjoint} factorization of
  polynomials in $\FX$ into \emph{variable-disjoint irreducibles}.  It
  turns out that such factorizations are unique and computing them is
  polynomial-time equivalent to polynomial identity testing (for both
  noncommutative arithmetic circuits and algebraic branching
  programs).

\item[$\bullet$] It turns out that we can apply the algorithm for
  variable-disjoint factorization to two special cases of
  factorization in $\FX$: homogeneous polynomials and multilinear
  polynomials. These polynomials do have unique factorizations and we
  obtain efficient algorithms for computing them.
\end{itemize}

$\bullet$ We also study a natural polynomial decomposition problem for
noncommutative polynomials and obtain complexity results.


\section{Variable-disjoint Factorization Problem}\label{two}

In this section we consider the problem of factorizing a
noncommutative polynomial $f\in\FX$ into variable disjoint factors.

For a polynomial $f\in\FX$ let $Var(f)\subseteq X$ denote the set of
all variables occurring in nonzero monomials of $f$.

\begin{definition}
A nontrivial \emph{variable-disjoint factorization} of a polynomial $f
\in \FX$ is a factorization
\[
f=gh
\]
such that $\deg g >0$ and $\deg h >0$, and $Var(g)\cap
Var(h)=\emptyset$.

A polynomial $f$ is \emph{variable-disjoint irreducible} if does not
have a nontrivial variable-disjoint factorization.

\end{definition}

Clearly, all irreducible polynomials are also variable-disjoint
irreducible. But the converse is not true. For instance, the familiar
polynomial $xyx+x$ is variable-disjoint irreducible but not
irreducible. Furthermore, all univariate polynomials in $\FX$ are
variable-disjoint irreducible.

We will study the complexity of variable-disjoint factorization for
noncommutative polynomials. First of all, it is interesting that
although we do not have the usual unique factorization in the ring
$\FX$, we can prove that every polynomial in $\FX$ has a \emph{unique}
variable-disjoint factorization into variable-disjoint irreducible
polynomials.\footnote{Uniqueness of the factors is upto scalar
  multiplication.}

We can exploit the properties we use to show uniqueness of
variable-disjoint factorization for computing the variable-disjoint
factorization. Given $f \in \FX$ as input by a noncommutative
arithmetic circuit the problem of computing arithmetic circuits for
the variable-disjoint irreducible factors of $f$ is polynomial-time
reducible to PIT for noncommutative arithmetic circuits. An analogous
result holds for $f$ given by an algebraic branching programs
(ABPs). Hence, there is a deterministic polynomial-time algorithm for
computing the variable-disjoint factorization of $f$ given by an ABP.
Also in the case when the polynomial $f\in\FX$ is given as input by a
black-box (appropriately defined) we give an efficient algorithm that
gives black-box access to each variable-disjoint irreducible factor of
$f$.

\begin{remark}
Factorization of \emph{commutative} polynomials into variable-disjoint
factors is studied by Shpilka and Volkovich in \rm{\cite{SV10}}. They
show a deterministic reduction to polynomial identity
testing. However, the techniques used in their paper are specific to
commutative rings, involving scalar substitutions, and do not appear
useful in the noncommutative case. Our techniques for factorization
are simple, essentially based on computing left and right partial
derivatives of noncommutative polynomials given by circuits or
branching programs.
\end{remark}

\subsection{Uniqueness of variable-disjoint factorization}

Although the ring $\FX$ is not a unique factorization domain we show
that factorization into variable-disjoint irreducible factors is
unique. 


For a polynomial $f\in\FX$ let $mon(f)$ denote the set of nonzero
monomials occurring in $f$.

\begin{lemma}\label{mondisjoint}
Let $f= gh$ such that $Var(g)\cap Var(h)= \phi$ and
$|Var(g)|,|Var(h)|\geq 1$.  Then 
\[
mon(f)= \{mw | m \in mon(g), w \in mon(h) \}.
\]
Moreover, the coefficient of $mw$ in $f$ is the product of the
coefficients of $m$ in $g$ and $w$ in $h$.
\end{lemma}

\begin{proof}
Let $m \in mon(g)$ and $w \in mon(h)$. We will argue that the monomial
$mw$ is in $mon(f)$ and can be obtained in a unique way in the product
$gh$. Namely, by multiplying $m \in mon(g)$ with $w \in mon(h)$.

Assume to the contrary that for some $u \in mon(g)$ and $v\in mon(h)$,
$u \neq m$ we have $mw= uv$. Note that $m$ and $u$ are words over
$Var(g)$ and $w,v$ are words over $Var(h)$. Clearly, $|m|\ne |u|$
because $m \ne u$. Without loss of generality, we assume that $|u|>
|m|$. As $mw= uv$, it follows that for some word $s\in X^*$, $|s|>0$,
we have $u= ms$ and $w=sv$.

That implies $Var(s) \subseteq Var(g)\cap Var(h)$ which contradicts
variable disjointness of $g$ and $h$.
\end{proof}

\begin{lemma}\label{var-subset}
Let $f=g.h$ and $f=u.v$ be two nontrivial variable-disjoint
factorizations of $f$. That is, 
\begin{eqnarray*}
Var(g)\cap Var(h) = \phi\\
Var(u)\cap Var(v) = \phi. 
\end{eqnarray*}
Then either $Var(g)\subseteq Var(u)$ and $Var(h)\supseteq Var(v)$ or
$Var(u)\subseteq Var(g)$ and $Var(v)\supseteq Var(h)$.
\end{lemma}

\begin{proof}
Suppose to the contrary that $x \in Var(g)\setminus Var(u)$ and $y \in
Var(u) \setminus Var(g)$. Let $m$ be any monomial of $f$ in which
variable $y$ occurs, and let $m=m_1m_2$ such that $m_1\in mon(g)$ and
$m_2\in mon(h)$. Clearly, $y$ must occur in $m_2$. Similarly, for any
monomial $m'$ of $f$ in which variable $x$ occurs, if $m'=m'_1m'_2$,
for $m'_1\in mon(g)$ and $m'_2\in mon(h)$, then $x$ must occur in
$m'_1$. Thus, $x\in Var(g)$ and $y\in Var(h)$. Similarly, it also
holds that $y\in Var(u)$ and $x\in Var(v)$. Thus we have:

\begin{enumerate}
\item $x \in Var(g)$ and $y \in Var(h)$ 
\item $y \in Var(u)$ and $x\in Var(v)$.
\end{enumerate}

Hence, there are monomials of $f$ in which both $x$ and $y$
occur. Furthermore, (1) above implies that in each monomial of $f$
containing \emph{both} $x$ and $y$, $x$ always occurs before $y$. On
the other hand, (2) implies $y$ occurs before $x$ in each such
monomial of $f$. This contradicts our assumption.
\end{proof}

\begin{lemma}\label{equaldisjoint}
Let $f \in \FX$ and suppose $f=gh$ and $f=uv$ are two
variable-disjoint factorizations of $f$ such that $Var(g)= Var(u)$.
Then $g=\alpha u$ and $h= \beta v$ for scalars $\alpha, \beta \in\F$.
\end{lemma}

\begin{proof}
As $f=gh=uv$ are variable-disjoint factorizations, by
Lemma~\ref{mondisjoint} each monomial $m\in mon(f)$ is uniquely
expressible as a product $m=m_1m_2$ for $m_1\in mon(g)$ and $m_2\in
mon(h)$, and as a product $m=m_1'm_2'$ for $m_1'\in mon(u)$ and
$m_2'\in mon(v)$. 

As $Var(g)=Var(u)$, we notice that $Var(h)=Var(v)$, because
$Var(f)=Var(g)\uplus Var(h)$ and $Var(f)=Var(u)\uplus Var(v)$.  Now,
from $m=m_1m_2=m_1'm_2'$ it immediately follows that $m_1=m_1'$ and
$m_2=m_2'$. Hence, $mon(g)=mon(u)$ and $mon(h)=mon(v)$. 

Furthermore, if $m$ is a monomial of maximum degree in $g$, by taking
the left partial derivative of $f$ w.r.t.\ $m$ we have
\[
\frac{\pd^\ell f}{\pd m}= \alpha' h = \beta' v,
\]
where $\alpha'$ and $\beta'$ are coefficient of $m$ in $g$ and $u$
respectively. It follows that $h= \beta v$ for some $\beta\in\F$.
Similarly, by taking the right partial derivative of $f$ w.r.t.\  a
maximum degree monomial in $h$ we can see that $g=\alpha u$ for
$\alpha\in\F$.
\end{proof}

We now prove the uniqueness of variable-disjoint factorizations in
$\FX$.

\begin{theorem}\label{thm-disjfact}
Every polynomial in $\FX$ has a unique variable-disjoint factorization
as a product of variable-disjoint irreducible factors, where the
uniqueness is upto scalar multiples of the irreducible factors.
\end{theorem}

\begin{proof}
 We prove the theorem by induction on the degree $d$ of polynomials in
 $\FX$. The base case $d=1$ is obvious, since degree $1$ polynomials
 are irreducible and hence also variable-disjoint irreducible. Assume
 as induction hypothesis that the theorem holds for polynomials of
 degree less than $d$. Let $f\in\FX$ be a polynomial of degree $d$. If
 $f$ is variable-disjoint irreducible there is nothing to prove.
 Suppose $f$ has nontrivial variable-disjoint factors. Let $f=gh$ be a
 \emph{nontrivial} variable-disjoint factorization such that $Var(g)$
 is \emph{minimum}. Such a factorization must exist because of Lemmas
 \ref{mondisjoint} and \ref{var-subset}. Furthermore, the set $Var(g)$
 is uniquely defined and by Lemma~\ref{equaldisjoint} the left factor
 $g$ is unique upto scalar multiples. Applying induction hypothesis to
 the polynomial $h$ now completes the induction step.
\end{proof}

 \subsection{Equivalence with PIT}\label{disjfact-pit}

\begin{theorem}\label{disjfactalgo}
Let $f\in\FX$ be a polynomial as input instance for variable-disjoint
factorization. Then
\begin{enumerate}
\item If $f$ is input by an arithmetic circuit of degree $d$ and size
  $s$ there is a randomized $\poly(s,d)$ time algorithm that
  factorizes $f$ into variable-disjoint irreducible factors.
\item
If $f$ is input by an algebraic branching program there is a
deterministic polynomial-time algorithm that factorizes $f$ into its
variable-disjoint irreducible factors.
\end{enumerate}
\end{theorem}

\begin{proof}
We first consider the most general case of the algorithm when $f \in
\FX$ is given by an arithmetic circuit of polynomially bounded
syntactic degree. The algorithm specializes to the other cases
too. The algorithm must compute an arithmetic circuit for each
variable-disjoint factor of $f$. We explain the polynomial-time
reduction to Polynomial Identity Testing (PIT) for noncommutative
arithmetic circuits.

Let $d=\deg f$ which can be computed by first computing circuits for
the polynomially many homogeneous parts of $f$, and then using PIT as
subroutine on each of them.

Next, we compute a monomial $m \in mon(f)$ of maximum degree $d$ with
polynomially many subroutine calls to PIT. The basic idea here is to
do a prefix search for $m$. More precisely, suppose we have computed a
prefix $m'$ such that the left partial derivative $\frac{\pd^\ell
  f}{\pd m'}$ computes a polynomial of degree $d-|m'|$ then we extend
$m'$ with the first variable $x_i\in X$ such that $\frac{\pd^\ell
  f}{\pd m'x}$ computes a polynomial of degree $d-|m'|-1$. Proceeding
thus, we will compute the lexicographically first monomial $m\in
mon(f)$ of degree $d$.

Now, starting with $d_1=1$ we look at all factorizations $m=m_1.m_2$
of monomial $m$ with $|m_1|=d_1$.

Let $h=\frac{\pd^\ell f}{\pd m_1}$ and $g=\frac{\pd^r f}{\pd m_2}$,
circuits for which can be efficiently computed from the given circuit
for $f$. Let $\alpha$, $\beta$ and $\gamma$ be the coefficients of $m$
in $f$, $m_1$ in $g$, and $m_2$ in $h$, respectively (which can be
computed in deterministic polynomial time \cite{AMS09}).

Next we compute $Var(g)$ and $Var(h)$. Notice that $x_i\not\in Var(g)$
if and only if $g$ and $g_{|_{x_i=0}}$ are not identical
polynomials. Thus, we can easily determine $Var(g)$ and $Var(h)$ with
$n$ subroutine calls to PIT. 

Clearly, if $Var(g)\cap Var(h) \neq \emptyset$ then we do not have a
candidate variable-disjoint factorization that splits $m$ as $m_1m_2$
and we continue with incrementing the value of $d_1$. Else, we check
if

\begin{equation}\label{checkfactor}
f={\frac{\alpha}{\beta\gamma}}gh,
\end{equation}

with a subroutine call to PIT. If $f={\frac{\alpha}{\beta\gamma}}gh$
then ${\frac{\alpha}{\beta\gamma}}g$ is the unique leftmost
variable-disjoint irreducible factor of $f$(upto scalar
multiplication), and we continue the computation. Otherwise, we
continue the search with incrementing the value of $d_1$.

In the general step, suppose we have already computed $f=g_1g_2\dots
g_i h_i$, where $g_1,g_2,\dots,g_i$ are the successive
variable-disjoint irreducible factors from the left. There will be
a corresponding factorization of the monomial $m$ as
\[
m=m_1m_2\dots m_i m_i'
\]
where $m_j$ occurs in $g_j$, $1\le j\le i$ and $m_i'$ occurs in $h_i$.
Notice that the polynomial $h_i = \frac{\pd^\ell f}{\pd m_1m_2\dots
  m_i}$ can be computed by a small noncommutative arithmetic circuit
obtained from $f$ by partial derivatives. This describes the overall
algorithm proving that variable-disjoint factorization of $f$ is
deterministic polynomial-time reducible to PIT when $f$ is given by
arithmetic circuits. The algorithm outputs arithmetic circuits for
each of the variable-disjoint irreducible factors of $f$.

We note that the algorithm specializes to the case when $f$ is given
by an ABP. The variable-disjoint irreducible factors are computed by
ABPs in this case.

Finally, to complete the proof we note that in the case of
noncommutative arithmetic circuits of polynomial degree there is a
randomized polynomial time PIT algorithm \cite{BW05}. For polynomials
given as ABPs, there is a deterministic polynomial-time PIT algorithm
\cite{RS05}.
\end{proof}

Next we consider variable-disjoint factorization of polynomials input
in sparse representation and show that the problem is solvable in
deterministic logspace (even by constant-depth circuits). Recall that
$\mathrm{AC}^0$ circuits mean a family of circuits $\{C_n\}$, where
  $C_n$ for length $n$ inputs, such that $C_n$ has polynomially
  bounded size and constant-depth and is allowed unbounded fanin AND
  and OR gates. The class of $\mathrm{TC}^0$ circuits is similarly
  defined, but is additionally allowed unbounded fanin majority gates.
  The logspace uniformity condition means that there is a logspace
  transducer that outputs $C_n$ on input $1^n$ for each $n$.

\begin{theorem}\label{disjfact-sparse}
Let $f\in\FX$ be a polynomial input instance for variable-disjoint
factorization given in sparse representation.
\begin{enumerate}
\item[(a)] When $\F$ is a fixed finite field the variable-disjoint
  factorization is computable in deterministic logspace (more
  precisely, even by logspace-uniform $\mathrm{AC}^0$ circuits).
\item[(b)] When $\F$ is the field of rationals the variable-disjoint
  factorization is computable in deterministic logspace (even by
  logspace-uniform $\mathrm{TC}^0$ circuits).
\end{enumerate}
\end{theorem}

\begin{proof}
We briefly sketch the simple circuit constructions. Let us first
consider the case when $\F$ is a fixed finite field. The idea behind
the $\mathrm{AC}^0$ circuit for it is to try all pairs of indices
$1\le a < b < d$ in parallel and check if $f$ has an irreducible
factor of degree $b-a+1$ located in the $(a,b)$ range. More precisely,
we will check in parallel for a variable-disjoint factorization
$f=g_1g_2g_3$, where $g_1$ is of degree $a-1$, $g_2$ is
variable-disjoint irreducible of degree $b-a+1$, and $g_3$ is of
degree $d-b$. If such a factorization exists then we will compute the
unique (upto scalar multiples) irreducible polynomial $g_3$ located in
the $(a,b)$ range of degrees.

In order to carry out this computation, we take a highest degree
nonzero monomial $m$ of $f$ and factor it as $m=m_1m_2m_3$, where
$m_2$ is located in the $(a,b)$ range of $m$. Clearly, the polynomial
\[
\frac{\pd^\ell}{\pd m_1}\left(\frac{\pd^r f}{\pd m_3}\right)
\]
has to be $g_2$ upto a scalar multiple. Likewise, $\frac{\pd^\ell
  f}{\pd m_1m_2}$ is $g_3$ and $\frac{\pd^r f}{\pd m_2m_3}$ is $g_1$
upto scalar multiples. Since $f$ is given in sparse representation, we
can compute these partial derivatives for each monomial in parallel in
constant depth.

We can then check using a polynomial size constant-depth circuit if
$f=g_1g_2g_3$ upto an overall scalar multiple (since we have assumed
$\F$ is a constant size finite field). For a given $a$, the least $b$
for which such a factorization occurs will give us the
variable-disjoint irreducible factor located in the $(a,b)$ range. For
each $a < d$ we can carry out this logspace computation looking for
the variable-disjoint irreducible factor (if there is one) located at
$(a,b)$ range for some $b>a$.

In this manner, the entire constant-depth circuit computation will
output all variable-disjoint irreducible factors of $f$ from left to
right. This concludes the proof for the constant size finite field
case. 

Over the field of rationals, we use the same approach as above. The
problem of checking if $f=g_1g_2g_3$ upto an overall scalar multiple
will require integer multiplication of a constant number of integers.
This can be carried out by a constant-depth threshold circuit. This
gives the $\mathrm{TC}^0$ upper bound.
\end{proof}

We now observe that PIT for noncommutative arithmetic circuits is also
deterministic polynomial-time reducible to variable-disjoint
factorization, making the problems polynomial-time equivalent.

\begin{theorem}
Polynomial identity testing for noncommutative polynomials $f\in\FX$
given by arithmetic circuits (of polynomial degree) is deterministic
polynomial-time equivalent to variable-disjoint factorization of
polynomials given by noncommutative arithmetic circuits.
\end{theorem}

\begin{proof}
Clearly proof of Theorem \ref{disjfactalgo} gives a reduction from variable-disjoint factorization to PIT. To see the other direction, let $f\in\FX$ and $y,z\not\in X$ be two new variables. We observe that the polynomial $f\in\FX$ is identically zero iff the polynomial $f+ y.z$ has a nontrivial variable-disjoint factorization. This gives a reduction from PIT to variable-disjoint factorization.
\end{proof}

\subsection{Black-Box Variable-Disjoint Factorization Algorithm}

In this subsection we give an algorithm for variable-disjoint
factorization when the input polynomial $f\in\FX$ is given by
black-box access. We explain the black-box model below:

In this model, the polynomial $f\in\FX$ can be evaluated on any
$n$-tuple of matrices $(M_1,M_2,\ldots,M_n)$ where each $M_i$ is a $t
\times t$ matrix over $\F$ (or a suitable extension field of $\F$) and
get the resulting $t \times t$ matrix $f(M_1,M_2,\ldots,M_n)$ as
output.

The algorithm for black-box variable-disjoint factorization takes as
input such a black-box access for $f$ and outputs black-boxes for each
variable-disjoint irreducible factor of $f$. More precisely, the
factorization algorithm, on input $i$, makes calls to the black-box
for $f$ and works as black-box for the $i^{th}$ variable-disjoint
irreducible factor of $f$, for each $i$.

The efficiency of the algorithm is measured in terms of the number of
calls to the black-box and the size $t$ of the matrices. In this
section we will design a variable disjoint factorization algorithm
that makes polynomially many black-box queries to $f$ on matrices of
polynomially bounded dimension. We state the theorem formally and then
prove it in the rest of this section.

\begin{theorem}\label{vardisj-thm}
Suppose $f\in\FX$ is a polynomial of degree bounded by $D$, given as
input via black-box access.  Let $f = f_1 f_2 \dots f_r$ be the
variable-disjoint factorization of $f$. Then there is a
polynomial-time algorithm that, on input $i$, computes black-box
access to the $i^{th}$ factor $f_i$.
\end{theorem}

\begin{proof}
In the proof we sketch the algorithm to give an overview. We then
explain the main technical parts in the next two lemmas. The black-box
algorithm closely follows the white-box algorithm already described in
the first part of this section. The main steps of the algorithm are:

\begin{enumerate}
\item We compute a nonzero monomial $m$ of highest degree $d\le D$ in
  $f$ by making black-box queries to $f$. We can also compute the
  coefficient $f(m)$ of $m$ by making black-box queries to $f$.

\item Then we compute the factorization of this monomial
\[
m = m_1 m_2 \dots m_r
\]
which, as in the white-box case, is precisely how the monomial $m$
factorizes according to the unique variable-disjoint factorization
$f=f_1 f_2 \dots f_r$ of $f$ (that is $m_i \in mon(f_i)$).

\item In order to find this factorization of the monomial $m$, it
  suffices to compute all factorizations $m=m'm''$ of $m$ such that
  $f=(\frac{\pd^\ell f}{\pd m'})(\frac{\pd^r f}{\pd m''})$ is a
  variable disjoint factorization (this is done using
  Equation~\ref{checkfactor}, exactly as in the proof of
  Theorem~\ref{disjfactalgo}).

\end{enumerate}

Finding all such factorizations will allow us to compute the
$m_i$. However, since we have only black-box access to $f$ we will
achieve this by creating and working with black-box access to the
partial derivatives $\frac{\pd^\ell f}{\pd m'}$ and $\frac{\pd^r
  f}{\pd m''})$ for each factorization $m=m'm''$. We explain
the entire procedure in detail.

Starting with $d_1=1$ we look at all factorizations $m=m'm''$ of
monomial $m$ with $|m'|=d_1$.

Suppose we have computed black-boxes for $h=\frac{\pd^\ell f}{\pd m'}$
and $g=\frac{\pd^r f}{\pd m''}$. Let $\alpha$, $\beta$ and $\gamma$ be
the coefficients of $m$ in $f$, $m'$ in $g$, and $m''$ in $h$, which
can be computed using black-box queries to $f$, $g$ and $h$,
respectively.

Next we compute $Var(g)$ and $Var(h)$. Notice that $x_i\not\in Var(g)$
if and only if $g$ and $g_{|_{x_i=0}}$ are not identical
polynomials. This amounts to performing PIT with black-box access to
$g$.  Thus, we can easily determine $Var(g)$ and $Var(h)$ by
appropriately using black-box PIT.

Clearly, if $Var(g)\cap Var(h) \neq \emptyset$ then we do not have a
candidate variable-disjoint factorization that splits $m$ as $m_1m_2$
and we continue with incrementing the value of $d_1$. Else, we check
if

\begin{equation}
f={\frac{\alpha}{\beta\gamma}}gh,
\end{equation}

with a subroutine call to PIT (like Equation~\ref{checkfactor} in
Theorem~\ref{disjfactalgo}). If $f={\frac{\alpha}{\beta\gamma}}gh$
then ${\frac{\alpha}{\beta\gamma}}g$ is the unique leftmost
variable-disjoint irreducible factor of $f$ (upto scalar
multiplication), and we continue the computation. Otherwise, we
continue the search with incrementing the value of $d_1$.

Suppose we have computed the monomial factorization $m=m_1m_2\dots m_i
m_i'$. There is a corresponding factorization of $f$ as $f=f_1f_2\dots
f_i h_i$, where $f_1,f_2,\dots,f_i$ are the successive
variable-disjoint irreducible factors from the left, and $m_j$ occurs
in $f_j$, $1\le j\le i$ and $m_i'$ occurs in $h_i$. Then we have
black-box access to $h_i = \frac{\pd^\ell f}{\pd m_1m_2\dots m_i}$ and
we can find its leftmost irreducible variable-disjoint factor as
explained above. This completes the overall algorithm sketch for
efficiently computing black-boxes for the irreducible
variable-disjoint factors of $f$ given by black-box.
\end{proof}

\begin{remark}
The variable-disjoint irreducible factors of $f$ are unique only upto
scalar multiples. However, we note that the algorithm in
Theorem~\ref{vardisj-thm} computes as black-box some \emph{fixed}
scalar multiple for each variable-disjoint irreducible factor
$f_i$. Since we can compute the coefficient of the monomial $m$ in $f$
(where $m$ is computed in the proof of Theorem~\ref{vardisj-thm}), we
can even ensure that the product $f_1f_2\dots f_r$ equals $f$ by
appropriate scaling.
\end{remark}

\begin{lemma}\label{bblemma1}
 Given a polynomial $f$ of degree $d$ with black-box access, we can
 compute a degree-$d$ nonzero monomial of $f$, if it exists, with at
 most $nd$ calls to the black-box on $(D+1)2d \times (D+1)2d$ matrices.
\end{lemma}

\begin{proof}
Similar to the white-box case, we will do a prefix search for a
nonzero degree-$d$ monomial $m$. We explain the matrix-valued queries
we will make to the black-box for $f$ in two stages. In the first
stage the query matrices will have the noncommuting variables from $X$
as entries.  In the second stage, which we will actually use to query
$f$, we will substitute matrices (with scalar entries) for the
variables occurring in the query matrices.

In order to check if there is a non-zero degree $d$ monomial in $f$
with $x_i$ as first variable, we evaluate $f$ on $(D+1) \times (D+1)$
matrices $X_1,\ldots,X_n$ defined as follows.
\[ X_1 = \left( \begin{array}{ccccc}
0 & 1 & 0 & \dots & 0\\
\vdots & 0 & x_1 & \ddots  & \vdots \\
\vdots & \ddots & \ddots &  \ddots & 0\\
\vdots & \ddots & \ddots & 0 & x_1 \\
0 & \ldots & \ldots & \ldots & 0 \\ 
\end{array} \right) \]
and for $i>1$, 
\[ X_i = \left( \begin{array}{ccccc}
0 & 0 & \ldots & \dots & 0\\
\vdots & 0 & x_i & \ddots  & \vdots \\
\vdots & \ddots & \ddots &  \ddots & 0\\
\vdots & \ddots & \ddots & 0 & x_i \\
0 & \ldots & \ldots & \ldots & 0  
\end{array} \right). \]

Therefore, $f(X_1,\dots,X_n)$ is a $(D+1) \times (D+1)$ matrix with
entries in $\FX$. In particular, by our choice of the matrices above,
the $(1,d+1)^{th}$ entry of $f(X_1,\dots,X_n)$ will precisely be
$\frac{\pd^\ell f_d}{\pd x_1}$, where $f_d$ is the degree $d$
homogeneous part of $f$. To check if $\frac{\pd^\ell f_d}{\pd x_1}$ is
nonzero, we can substitute random $t \times t$ matrices for each
variable $x_i$ occurring in the matrices $X_i$ and evaluate $f$ to
obtain a $((D+1)t)\times((D+1)t)$ matrix. In this matrix
$(1,d+1)^{th}$ block of size $t \times t$ will be nonzero with high
probability for $t= 2d$ \cite{BW05}. In general, suppose $m'$ is a
degree $d$ monomial that is nonzero in $f$. We search for the first
variable $x_i$ such that $m'x_i$ is a prefix of some degree $d$
nonzero monomial by similarly creating black-box access to
$\frac{\pd^\ell f_d}{\pd m'x_i}$ and then substituting random $t
\times t$ matrices for the variables and evaluating $f$ to obtain a
$((D+1)t)\times((D+1)t)$ matrix. The $(1,d+1)^{th}$ block of size $t
\times t$ is nonzero with high probability \cite{BW05} if $m'x_i$ is a
prefix for some nonzero degree $d$ monomial. Continuing thus, we can
obtain a nonzero monomial of highest degree $d$ in $f$.
\end{proof}

We now describe an efficient algorithm that creates black box access
for left and right partial derivatives of $f$ w.r.t.\ monomials.  Let
$m\in X^*$ be a monomial. We recall that the left partial derivative
$\frac{\pd^\ell f}{\pd m}$ of $f$ w.r.t.\ $m$ is the polynomial
\[
\frac{\pd^\ell f}{\pd m} = \sum_{f(mm')\ne 0} f(mm')m'.
\]

Similarly, the right partial derivative $\frac{\pd^r f}{\pd m}$ of $f$
w.r.t.\ $m$ is the polynomial
\[
\frac{\pd^r f}{\pd m} = \sum_{f(m'm)\ne 0} f(m'm)m'.
\]

\begin{lemma}\label{bblemma2}
 Given a polynomial $f$ of degree $d$ with black-box access, there are
 efficient algorithms that give black-box access to the polynomials
 $\frac{\pd^\ell f}{\pd m_1}$, $\frac{\pd^r f}{\pd m_2}$ for any
 monomials $m_1,m_2\in X^*$. Furthermore, there is also an efficient
 algorithm giving black-box access to the polynomial
 $\frac{\pd^\ell}{\pd m_1}\left(\frac{\pd^r f}{\pd m_2}\right)$.
\end{lemma}

\begin{proof}
Given the polynomial $f$ and a monomial $m_1$, we first show how to
create black-box access to $\frac{\pd^\ell f}{\pd m_1}$.  Let the
monomial $m_1$ be $x_{i_1}\dots x_{i_k}$.  Define matrices
$X_1,\dots,X_n$ as follows:



\[ X_i = \left( \begin{array}{cccc|cccc}
0 & T_i[1] & 0 & \dots & 0 & \ldots& \ldots& 0\\
\vdots & 0 & T_i[2]  & \ddots & \ddots  & \vdots & \ldots&\vdots  \\
\vdots & \ddots & 0 &  T_i[k] & 0 & \ldots& \ldots& 0 \\
\hline
0 & \ldots & \ldots & 0 &  x_i & \dots & \ldots&  0\\
\vdots & \ldots & \ldots & \vdots & 0 & x_i  & 0  & 0 \\ 
\vdots & \ldots & \ldots & \vdots & \vdots& \ddots  & \ddots & \vdots \\
\vdots & \ldots & \ldots & \vdots & \vdots& \ddots&  \ddots & x_i\\
0 & \ldots & \ldots & 0 & 0& \dots  & \ldots& 0
\end{array} \right). \]

For $1\le r\le k$, the entry $T_i[r]=1$ if $x_{i_r}=x_i$ and
$T_i[r]=0$ otherwise.

The black-box access to $\frac{\pd^\ell f}{\pd m_1}$ on input matrices
$M_1,\dots,M_n$ of size $t \times t$ can be created as follows. Note
that in the $(D+1)\times (D+1)$ matrix $f(X_1,\dots,X_n)$, the
$(1,j+1)^{th}$ location contains $\frac{\pd^\ell f_j}{\pd m_1}$, where
$f_j$ is the degree $j^{th}$ homogeneous part of $f$. Now, suppose for
each variable $x_i$ in the matrix $X_i$ we substitute a $t\times t$
scalar-entry matrix $M_i$ and compute the resulting $f(X_1,\dots,X_n)$
which is now a $(D+1) \times (D+1)$ block matrix whose entries are
$t\times t$ scalar-entry matrices.  Then the block matrix located at
the $(1,j+1)^{st}$ entry for $j\in \{2,\dots,d+1\}$ is the evaluation
of $\frac{\pd^\ell f_j}{\pd m_1}$ on $M_1,\dots,M_n$. We output the
sum of these matrix entries over $2\le j\le d+1$ as the black-box
output. \\

Next, we show how to create black-box access to $\frac{\pd^r f}{\pd
  m_2}$. Let the monomial $m_2 = x_{i_1} x_{i_2} \dots x_{i_k}$.  We
will define matrices $X_i^{(j)}$ for $i\in [n], j \in \{k,k+1\dots,D
\}$ as follows:

\[ X^{(j)}_i = \left( \begin{array}{cccc|cccc}
0 & x_{i} & 0 & \dots & 0 & \ldots& \ldots& 0\\
\vdots &   \vdots  & \ddots & \ddots  & \vdots & \ldots& \ldots&\vdots  \\
\vdots & \ddots & 0 &  x_{i} & 0 & \ldots& \ldots& 0 \\
\hline
0 & \ldots & \ldots & 0 &  T_{i}[1] & \dots & \ldots&  0\\
\vdots & \ldots & \ldots & \vdots & 0 & T_{i}[2]  & 0  & 0 \\ 
\vdots & \ldots & \ldots & \vdots & \vdots& \ddots  & \ddots & \vdots \\
\vdots & \ldots & \ldots & \vdots & \vdots& \ddots&  \ddots & T_{i}[k]\\
0 & \ldots & \ldots & 0 & 0& \dots  & \ldots& 0
\end{array} \right). \]

The matrix $X_i^{(j)}$ is a $(j+1)\times (j+1)$ matrix where, in the
figure above the top left block is of dimension $(j-k)\times (j-k+1)$.
Here $T_{i}[r]=1$ if $x_{i_r}=x_i$ and $T_{i}[r]=0$ otherwise.

Finally, define the block diagonal matrix $X_i =
diag(X_i^{(k)},\dots,X_i^{(D)})$.

We now describe the black-box access to $\frac{\pd^r f}{\pd m_2}$.  Let
$x_i\leftarrow M_i$ be the input assignment of $t\times t$ matrices
$M_i$ to the variables $x_i$. This results in matrices $\hat{X_i},
1\le i\le n$, where $\hat{X_i}$ is obtained from $X_i$ by replacing
$x_i$ by $M_i$. 

The query $f(\hat{X_1},\hat{X_2},\ldots,\hat{X_n})$ to $f$ gives a
block diagonal matrix $diag(N_k,N_{k+1},\ldots,N_D)$. Here the matrix
$N_j$ is a $(j+1)\times (j+1)$ block matrix with entries that are
$t\times t$ matrices (over $\F$). The $(1,j+1)^{th}$ block in $N_j$ is
$\frac{\pd^r f_j}{\pd m_2}$ evaluated at $M_1,M_2,\ldots,M_n$.

Hence the sum of the $(1,j+1)^{th}$ blocks of the different $N_j, k\le
k\le D$ gives the desired black-box access to $\frac{\pd^r f_j}{\pd
 m_2}$.

\end{proof}

\section{Factorization of Multilinear and Homogeneous Polynomials}\label{three}

In this section we briefly discuss two interesting special cases of
the standard factorization problem for polynomials in $\FX$. Namely,
the factorization of multilinear polynomials and the factorization of
homogeneous polynomials. It turns out, as we show, that factorization
of multilinear polynomials coincides with their variable-disjoint
factorization. In the case of homogeneous polynomials, it turns out
that by renaming variables we can reduce the problem to
variable-disjoint factorization.

In summary, multilinear polynomials as well as homogeneous polynomials
have unique factorizations in $\FX$, and by the results of the
previous section these can be efficiently computed.

A polynomial $f\in\FX$ is \emph{multilinear} if in every nonzero
monomial of $f$ every variable in $X$ occurs at most once. We begin by
observing some properties of multilinear polynomials.  Let $\var(f)$
denote the set of all indeterminates from $X$ which appear in some
nonzero monomial of $f$.
 
It turns out that factorization and variable-disjoint factorization of
multilinear polynomials coincide.

\begin{lemma}\label{lem-inter}
Let $f\in\FX$ be a multilinear polynomial and $f=gh$ be any nontrivial
factorization of $f$. Then, $Var(g) \cap Var(h) = \emptyset$.
\end{lemma}

\begin{proof}
Suppose $x_i \in Var(g) \cap Var(h)$ for some $x_i\in X$. Let $m_1$ be
a monomial in $g$ of maximal degree which also has the indeterminate
$x_i$ occurring in it. Similarly, let monomial $m_2$ be of maximal
degree in $h$ with $x_i$ occurring in it. The product monomial $m_1
m_2$ is not multilinear and it cannot be nonzero in $f$. This monomial
must therefore be cancelled in the product $gh$, which means there are
nonzero monomials $m_1'$ of $g$ and $m_2'$ of $h$ such that
$m_1m_2=m_1'm_2'$. Since $\deg(m_1')\le \deg(m_1)$ and $\deg(m_2')\le
\deg(m_2)$ the only possibility is that $m_1=m_1'$ and $m_2=m_2'$
which means the product monomial $m_1m_2$ has a nonzero coefficient in
$f$, contradicting the multilinearity of $f$. This completes the
proof.
\end{proof}

Thus, by Theorem~\ref{thm-disjfact} multilinear polynomials in $\FX$
have unique factorization. Furthermore, the algorithms described in
Section~\ref{disjfact-pit} can be applied to efficiently factorize
multilinear polynomials.

We now briefly consider factorization of homogeneous polynomials in
$\FX$. 

\begin{definition}
A polynomial $f\in\FX$ is said to be \emph{homogeneous} of degree $d$
if every nonzero monomial of $f$ is of degree $d$.
\end{definition}

Homogeneous polynomials do have the unique factorization
property. This is attributed to J.H.~Davenport in \cite{Caruso}.
However, we argue this by reducing the problem to variable-disjoint
factorization.

Given a degree-$d$ homogeneous polynomial $f\in\FX$, we apply the
following simple transformation to $f$: For each variable $x_i\in X$
we introduce $d$ variables $x_{i1},x_{i2},\ldots,x_{id}$. For each
monomial $m\in mon(f)$, we replace the occurrence of variable $x_i$ in
the $j^{th}$ position of $m$ by variable $x_{ij}$. The new polynomial
$f'$ is in $\F\angle{\{x_{ij}\}}$. The crucial property of homogeneous
polynomials we use is that for any factorization $f=gh$ both $g$ and
$h$ must be homogeneous.

\begin{lemma}\label{l:uniq1}
Let $f\in\FX$ be a homogeneous degree $d$ polynomial and $f'$ be the
polynomial in $\F\angle{\{x_{ij}\}}$ obtained as above. Then
\begin{itemize}
\item The polynomial $f'$ is variable-disjoint irreducible iff
$f$ is irreducible.
\item If $f'=g_1'g_2'\dots g_t'$ is the variable-disjoint
  factorization of $f'$, where each $g_k'$ is variable-disjoint
  irreducible then, correspondingly $f=g_1g_2\dots g_t$ is a
  factorization of $f$ into irreducibles $g_k$, where $g_k$ is
  obtained from $g_k'$ by replacing each variable $x_{ij}$ in $g_k'$
  by $x_i$.
\end{itemize}
\end{lemma}

\begin{proof}
The first part follows because if $f$ is reducible and $f=gh$ then
$f'=g'h'$, where $g'$ is obtained from $g$ by replacing the variables
$x_i$ by $x_{ij}$, and $h'$ is obtained from $h$ by replacing the
variables $x_i$ by $x_{i,j+s}$, where $s=\deg g$.

For the second part, consider the product $g_1g_2\dots g_t$. As all
the factors $g_k$ are homogeneous, it follows that each $g_k$ is
irreducible for otherwise $g_k$ is not variable disjoint irreducible.
Furthermore, any monomial $m$ in $g_1g_2\dots g_t$ can be uniquely
expressed as $m=m_1m_2\dots m_t$, where $m_k\in mon(g_k)$ for each
$k$. Thus, for each $i$ and $j$, replacing the $j^{th}$ occurrence of
$x_i$ by $x_{ij}$ in the product $g_1g_2\dots g_t$ will give us $f'$
again. Hence $g_1g_2\dots g_t$ is $f$.
\end{proof}

It follows easily that factorization of homogeneous polynomials is
reducible to variable-disjoint factorization and we can solve it
efficiently using Theorems~\ref{disjfactalgo} and
\ref{disjfact-sparse}, depending on how the polynomial is input. We
summarize this formally.

\begin{theorem}\label{uf-homogeneous}
Homogeneous polynomials $f\in\FX$ have unique factorizations into
irreducible polynomials. Moreover, this factorization can be
efficiently computed:
\begin{itemize}
\item Computing the factorization of a homogeneous polynomial $f$
  given by an arithmetic circuit of polynomial degree is
  polynomial-time reducible to computing the variable-disjoint
  factorization of a polynomial given by an arithmetic circuit.
\item Factorization of $f$ given by an ABP is constant-depth reducible
  to variable-disjoint factorization of polynomials given by ABPs.
\item Factorization of $f$ given in sparse representation is
  constant-depth reducible to variable-disjoint factorization of
  polynomials given by sparse representation.
\end{itemize}
\end{theorem}

\section{A Polynomial Decomposition Problem}\label{four}
  

Given a degree $d$ homogeneous noncommutative polynomial $f \in \FX$,
a number $k$ in unary as input we consider the following decomposition
problem, denoted by $\SOP$ (for sum of products decomposition): 

Does $f$ admit a decomposition of the form
\[
f = g_1 h_1 + \dots + g_k h_k? 
\]
where each $g_i\in\FX$ is a homogeneous polynomial of degree $d_1$ and
each $h_i\in\FX $ is a homogeneous polynomial of degree $d_2$. Notice
that this problem is a generalization of homogeneous polynomial
factorization. Indeed, homogeneous factorization is simply the case
when $k=1$.

\begin{remark}
As mentioned in \cite{Arnab14}, it is interesting to note that for
\emph{commutative polynomials} the complexity of \textsc{SOP} is open
even in the case $k=2$.  However, when $f$ is of constant degree then
it can be solved efficiently by applying a very general algorithm
\cite{Arnab14} based on a regularity lemma for polynomials.
\end{remark}

When the input polynomial $f$ is given by an arithmetic circuit, we
show that \textsc{SOP} is in $\MA\cap\coNP$. On the other hand, when
$f$ is given by an algebraic branching program then $\textsc{SOP}$ can
be solved in deterministic polynomial time by some well-known
techniques (we can even compute ABPs for the $g_i$ and $h_i$ for the
minimum $k$).

\begin{theorem}\label{2sop}
Suppose a degree $d$ homogeneous noncommutative polynomial $f\in\FX$,
and positive integer $k$ encoded in unary are the input to $\SOP$:
\begin{itemize}
\item[(a)] If $f$ is given by a polynomial degree arithmetic circuit
  then $\SOP$ is in $\MA \cap \coNP$.
\item[(b)] If $f$ is given by an algebraic branching program then
  $\SOP$ is in deterministic polynomial time (even in randomized
  $\NC^2$).
\item[(c)] If $f$ is given in the sparse representation then
  \textsc{SOP} is equivalent to the problem of checking if the rank of
  a given matrix is at most $k$. In particular, if $\F$ is the field
  of rationals, $\SOP$ is complete for the complexity class
  $\mathrm{C_{=}L}$.\footnote{The logspace counting class
    $\mathrm{C_{=}L}$ captures the complexity of matrix rank over
    rationals~\cite{ABO99}.}
\end{itemize}
\end{theorem}

We first focus on proving part (a) of the theorem. If $(f,k)$ is a
``yes'' instance to \textsc{SOP}, then we claim that there exist small
arithmetic circuits for the polynomials $g_i,h_i$, $i \in [k]$.

We define the \emph{partial derivative matrix} $A_f$ for the
polynomial $f$ as follows. The rows of $A_f$ are indexed by degree
$d_1$ monomials and the columns of $A_f$ by degree $d_2$ monomials
(over variables in $X$). For the row labeled $m$ and column labeled
$m'$, the entry $A_{m,m'}$ is defined as
\[
A_{m,m'} = f(mm').
\]

The key to analyzing the decomposition of $f$ is the rank of the
matrix $A_f$.

\begin{claim}\label{2sopclm}
Let $f\in\FX$ be a homogeneous degree $d$ polynomial.

\begin{itemize}
\item[(a)] Then $f$ can be decomposed as $f = g_1 h_1 + \dots +g_k
  h_k$ for homogeneous degree $d_1$ polynomials $g_i$ and homogeneous
  degree $d_2$ polynomials $h_i$ if and only if the rank of $A_f$ is
  bounded by $k$.
\item[(b)] Furthermore, if $f$ is computed by a noncommutative
  arithmetic circuit $C$ then if the rank of $A_f$ is bounded by $k$
  there exist polynomials $g_i,h_i\in\FX$, $i\in [k]$, such that
  $f=g_1 h_1 + \dots +g_k h_k$, where $g_i$ and $h_i$ have
  noncommutative arithmetic circuits of size $poly(|C|,n,k)$
  satisfying the above conditions.
\end{itemize}
\end{claim}

\begin{proofof}{Claim~\ref{2sopclm}}

For a homogeneous degree $d_1$ polynomial $g$ let $\bar{g}$ denote its
coefficient column vector whose rows are indexed by degree $d_1$
monomials exactly as the rows of $A_f$.  Similarly, for a homogeneous
degree $d_2$ polynomial $h$ there is a coefficient row vector
$\bar{[h}$ with columns indexed by degree $d_2$ monomials (as the
 columns of $A_f$).

Observe that if $f$ can be decomposed as the sum of products $f=g_1h_1
+ \dots +g_k h_k$, where each $g_i$ is degree $d_1$ homogeneous and
each $h_i$ is degree $d_2$ homogeneous then the matrix $A_f$ can be
decomposed into a sum of $k$ rank-one matrices:

\[
A_f = \bar{g_1} \bar{h_1}^T + \dots + \bar{g_k} \bar{h_k}^T.
\]

It follows that the rank of $A_f$ is bounded by $k$. Conversely, if
the rank of $A_f$ is $k$ then $A_f$ can be written as the sum of $k$
rank $1$ matrices. Since each rank one matrix is of the form
$\bar{g}\bar{h}^T$, we obtain an expression as above which implies the
decomposition of $f$ as $g_1h_1+\dots + g_kh_k$.

Now, if the rank of $A_f$ is $k$ then there are degree $d_1$ monomials
$m_1,\dots,m_k$ and degree $d_2$ monomials $m_1',\dots,m_k'$ such that
the $k \times k$ minor of $A_f$ corresponding to these rows and
columns is an invertible matrix $K$. W.l.o.g, we can write the $p
\times q$ matrix $A_f$ as
\[
\left( \begin{array}{ccc}
K & \Delta_2 \\
\Delta_1 & J \end{array} \right)
\]
for suitable matrices $\Delta_1,\Delta_2,J$. Moreover, since $A_f$ is
rank $k$, we can row-reduce to obtain
\[
\left( \begin{array}{ccc}
I & O \\
-\Delta_1 K^{-1} & I \end{array} \right)
\left( \begin{array}{ccc}
K & \Delta_2 \\
\Delta_1 & J \end{array} \right) = 
\left( \begin{array}{ccc}
K & \Delta_2 \\
O & O \end{array} \right)
\]
and column reduce to obtain 
\[
\left( \begin{array}{ccc}
I & O \\
-\Delta_1 K^{-1} & I \end{array} \right)
\left( \begin{array}{ccc}
K & \Delta_2 \\
\Delta_1 & J \end{array} \right) 
\left( \begin{array}{ccc}
I & -K^{-1}\Delta_2 \\
O & I \end{array} \right)
= 
\left( \begin{array}{ccc}
K & O \\
O & O \end{array} \right)
\]
It is easy to verify that this yields the following factorization for
$A_f = U I_k V$
\[
A_f = \left( \begin{array}{ccc} K & O \\ \Delta_1 & I \end{array}
\right) \left( \begin{array}{ccc} I & O \\ O & O \end{array} \right)
\left( \begin{array}{ccc} I & K^{-1}\Delta_2 \\ O & I \end{array}
\right)
\]
Since we can write $I_k$ as the sum $e_1e_1'^T + \dots +e_k e_k'^T$
for standard basis vectors of suitable dimensions, we can express
$A_f$ as the sum of $k$ rank-one matrices $(Ue_i)(e_i'^TV)$.  We
observe that the column vector $Ue_i$ is the $i^{th}$ column of matrix
$U$, which is identical to the $i^{th}$ column of matrix $A_f$.
Therefore, the polynomial represented by this vector corresponds to
the partial derivative of $f$ w.r.t the monomial $m_i$, which can
computed by a circuit of desired size.  Similarly, any row vector
$e_i'^T V$ is the $i^{th}$ row of matrix $V$, which is identical to
the $i^{th}$ row of matrix $A_f$, scaled by the transformation
$K^{-1}$. Moreover, the entries of $K^{-1}$ are efficiently
computable, and have size bounded by polynomial in the input size for
the fields of our interest.  Therefore, the polynomial represented by
this vector corresponds to a linear combination of the partial
derivatives of $f$ w.r.t the monomials $m_1',\dots,m_k'$, which can
computed by a circuit of desired size. Therefore, there exist
polynomials $g_i,h_i$ for $i \in [k]$ which satisfy the conditions
of the second part of the claim.
\end{proofof}

\begin{proofof}{Theorem \ref{2sop}}

\vspace{2mm}

\noindent\textbf{Part (a).} We first show that $\SOP$ is in
$\MA$. Given as input polynomial $f\in\FX$ by an arithmetic circuit,
the \textrm{MA} protocol works as follows.  Merlin sends the
description of arithmetic circuits for the polynomials $g_i, h_i, 1\le
i\le k$. By the second part of Claim \ref{2sopclm}, this message is
polynomial sized as the polynomials $g_i$ and $g_i$ have circuits of
size $poly(|C|,n,k)$. Arthur verifies that
$f=g_1h_1+g_2h_2+\dots+g_kh_k$ using a randomized noncommutative PIT
algorithm with suitable success probability. This establishes the
\textrm{MA} upper bound.

Next, we show that the complement problem $\overline{\SOP}$ is in NP.
I.e.\ given as input a polynomial $f\in\FX$ such that $f$ is not in
$\SOP$, we show that there is a $poly(|C|,n,k)$-sized proof,
verifiable in polynomial time, that $f$ cannot be written as $f =
g_1h_1 + \dots + g_k h_k$. It suffices to exhibit a short proof of the
fact that the rank of matrix $A_f$ is at least $k+1$. This can be done
by listing $k+1$ degree $d_1$ monomials $m_1,\dots,m_{k+1}$ and degree
$d_2$ monomials $m_1',\dots,m_{k+1}'$ such that the $(k+1)\times
(k+1)$ minor of $A_f$ indexed by these monomials has full rank. This
condition can be checked in polynomial time by verifying the
determinant of the minor to be non-zero.\\

\noindent\textbf{Part (b).} When the polynomial is given as an ABP
$P$, we sketch the simple polynomial-time algorithm for $\SOP$.

Our goal is to compute a decomposition 

\begin{equation}\label{decomp-eq}
f = g_1h_1 + \dots +g_kh_k
\end{equation}

for minimum $k$, where $\deg(g_i)=d_1$ and $\deg(h_i)=d_2$. In order
to compute this decomposition, we can suitably adapt the multiplicity
automaton learning algorithm of Beimel et al \cite{beim}. We can
consider the input homogeneous degree-$d$ polynomial $f\in\FX$ as a
function $f:X^d\rightarrow \F$, where $f(m)$ is the coefficient of
$m$ in the polynomial $f$. The learning algorithm works in Angluin's
model. More precisely, when given black-box access to the function
$f(m)$ for monomial queries $m$, it uses equivalence queries with
counterexamples and learns the minimum size ABP computing $f$ in
polynomial time. Given a hypothesis ABP $P'$ for $f$, we can simulate
the equivalence queries and finding a counterexample by using the
Raz-Shpilka PIT algorithm \cite{RS05} on $P-P'$. The minimized ABP that
is finally output by the learning algorithm will have the optimal
number of nodes at each layer. In particular, layer $d_1$ will have
the minimum number of nodes $k$ which will give the decomposition in
Equation~\ref{decomp-eq}. Furthermore, the ABPs for the polynomials
$g_i$ and $h_i$ can also be immediately obtained from the minimum size
ABP.\\

\noindent\textbf{Part (c).} When the polynomial $f\in\FX$ is given in
sparse representation, we can explicitly write down the partial
derivative matrix $A_f$ and check its rank. Moreover, we can even
compute the decomposition by computing a rank-one decomposition of the
partial derivative matrix.  The equivalence arises from the fact that
the rank of a matrix $A$ is equal to the minimum number of summands in
the decomposition of the noncommutative polynomial
$\displaystyle\sum_{i,j\in[n]} a_{ij}x_ix_j$ as a sum of products of
homogeneous linear forms.

\end{proofof}

\subsubsection*{An $\NP$-hard decomposition problem}

We now briefly discuss a generalization of $\SOP$. Given a polynomial
$f\in\FX$ as input along with $k$ in unary, can we decompose it as a
$k$-sum of products of \emph{three} homogeneous polynomials:
\[
f = a_1b_1c_1 + a_2b_2c_2 + \dots + a_kb_kc_k,
\]

where each $a_i$ is degree $d_1$, each $b_i$ is degree $d_2$, and
each $c_i$ is degree $d_3$?

It turns out that even in the simplest case when $f$ is a cubic
polynomial and the $a_i, b_i, c_i$ are all homogeneous linear forms,
this problem is NP-hard. The tensor rank problem: given a
$3$-dimensional tensor $A_{ijk}, 1\le i,j,k\le n$ checking if the
tensor rank of $A$ is bounded by $k$, which is known to be NP-hard
\cite{Has90} is easily shown to be polynomial-time reducible to this
decomposition problem.

Indeed, we can encode a three-dimensional tensor $A_{ijk}$ as a
homogeneous cubic noncommutative polynomial
$f=\sum_{i,j,k\in[n]}A_{ijk}x_iy_jz_k$, such that any summand in the
decomposition, which is product of three homogeneous linear forms,
corresponds to a rank-one tensor. This allows us to test whether a
tensor can be decomposed into at most $k$ rank-one tensors, which is
equivalent to testing whether the rank of the tensor is at most $k$.

\section{Concluding Remarks and Open Problems}\label{five}

The main open problem is the complexity of noncommutative polynomial
factorization in the general case. Even when the input polynomial
$f\in\FX$ is given in sparse representation we do not have an
efficient algorithm nor any nontrivial complexity-theoretic upper
bound. Although polynomials in $\FX$ do not have unique factorization,
there is interesting structure to the factorizations
\cite{cohn,cohn-ufd} which can perhaps be exploited to obtain
efficient algorithms.

In the case of irreducibility testing of polynomials in $\FX$ we have
the following observation that contrasts it with commutative
polynomials.  Let $\F$ be a fixed finite field. We note that checking
if $f\in\FX$ given in sparse representation is \emph{irreducible} is
in coNP. To see this, suppose $f$ is $s$-sparse of degree $D$. If $f$
is reducible and $f=gh$ is any factorization then each monomial in $g$
or $h$ is either a prefix or a suffix of some monomial of $f$. Hence,
both $g$ and $h$ are $sD$-sparse polynomials. An NP machine can guess
$g$ and $h$ (since coefficients are constant-sized) and we can verify
if $f=gh$ in deterministic polynomial time.

On the other hand, it is an interesting contrast to note that given an
$s$-sparse polynomial $f$ in the \emph{commutative} ring
$\F[x_1,x_2,\ldots,x_n]$ we do not know if checking irreducibility is
in coNP. However, checking irreducibility is known to be in RP
(randomized polynomial time with one sided-error) as a consequence of
the Hilbert irreducibility criterion \cite{Kalt}. If the polynomial $f$
is irreducible, then if we assign random values from a suitably large
extension field of $\F$ to variables $x_2,\ldots,x_n$ (say,
$x_i\leftarrow r_i$) the resulting univariate polynomial
$f(x_1,r_2,\ldots,r_n)$ is irreducible with high probability.

Another interesting open problem that seems closely related to
noncommutative sparse polynomial factorization is the problem of
finite language factorization \cite{SY99}. Given as input a finite
list of words $L=\{w_1,w_2,\dots,w_s\}$ over the alphabet $X$ the
problem is to check if we can factorize $L$ as $L=L_1L_2$, where $L_1$
and $L_2$ are finite sets of words over $X$ and $L_1L_2$ consists of
all strings $uv$ for $u\in L_1$ and $v\in L_2$. This problem can be
seen as noncommutative sparse polynomial factorization problem where
the coefficients come from the \emph{Boolean ring} $\{0,1\}$. No
  efficient algorithm is known for this problem in general, neither is
  any nontrivial complexity bound known for it \cite{SY99}. On the
  other hand, analogous to factorization in $\FX$, we can solve it
  efficiently when $L$ is \emph{homogeneous} (i.e.\ all words in $L$
  are of the same length). Factorizing $L$ as $L_1L_2$, where $L_1$
  and $L_2$ are variable-disjoint can also be efficiently done by
  adapting our approach from Section~\ref{two}. It would be
  interesting if we can relate language factorization to sparse
  polynomial factorization in $\FX$ for a field $\F$. Is one
  efficiently reducible to the other?


\begin{thebibliography}{}

\bibitem[ABO99]{ABO99}{\sc E. Allender, R. Beals, M. Ogihara,} The
  complexity of matrix rank and feasible systems of linear
  equations. {\sl Computational Complexity,} 8:2, 99-126, 1999.


\bibitem[Ar14]{Arnab14} {\sc Arnab Bhattacharya,} Polynomial
  Decompositions in Polynomial Time. {\sl Proceedings of the ESA
    Conference,} pages 125-136, 2014.


\bibitem[AMS10]{AMS09} {\sc V. Arvind, P. Mukhopadhyay,
  S. Srinivasan,} New Results on Noncommutative and Commutative
  Polynomial Identity Testing. {\sl Computational Complexity,}
  19(4):521-558 (2010).

\bibitem[BB+00]{beim} {\sc A. Beimel, F. Bergadano, N.q H. Bshouty,
  E. Kushilevitz, S. Varricchio,} Learning functions represented as
  multiplicity automata. {\sl Journal of the ACM,} 47(3): 506-530
  (2000).

\bibitem[BW05]{BW05} {\sc A. Bogdanov, H. Wee} More on Noncommutative
  Polynomial Identity Testing {\sl In Proc. of 20th Annual Conference
    on Computational Complexity,} 92-99, 2005.

\bibitem[Ca10]{Caruso} {\sc F. Caruso,} Factorization of
  Noncommutative Polynomials, {\sl CoRR abs/1002.3180}, 2010.

\bibitem[Co85]{cohn} {\sc P. M. Cohn,} {\sl Free rings and their
  relations,} Academic Press, London Mathematical Society Monograph
  No. 19, 1985.

\bibitem[Co]{cohn-ufd} {\sc P. M. Cohn,} Noncommutative unique
  factorization domains, {\sl Transactions of the American
    Math. Society,} 313-331: 109 (1963). 


\bibitem[GG]{vzGbook} {\sc J. Gathen, J. Gerhard,} Modern Computer
  Algebra $2^{nd}$ edition, {\sl Cambridge University Press.}


\bibitem[Has90]{Has90}{\sc Johan Hastad,}
Tensor rank is NP-complete. {\sl Journal of Algorithms},
11(4):644-654,1990.

\bibitem[Ka89]{Kalt} {\sc E. Kaltofen,} Factorization of Polynomials
  given by straight-line programs, {\sl Randomness in Computation,
    vol. 5 of Advances in Computing Research}, 375-412, 1989.

\bibitem[KT90]{KT91} {\sc E. Kaltofen and B. Trager,} Computing with
  polynomials given by black-boxes for their evaluations: Greatest
  common divisors, factorization, separation of numerators and
  denominators.  \emph{J. Symbolic Comput.,} 9(3):301-320, 1990.

\bibitem[KSS14]{KSS} {\sc S. Kopparty, S. Saraf,
  A. Shpilka,}Equivalence of Polynomial Identity Testing and
  Deterministic Multivariate Polynomial, {\sl Electronic Colloquium on
    Computational Complexity (ECCC) 21:1}, 2014.


\bibitem[N91]{Ni91} {\sc N. Nisan,} Lower bounds for noncommutative
  computation {\sl In Proc. of 23rd ACM Sym. on Theory of Computing,}
  410-418, 1991.

\bibitem[RS05]{RS05} {\sc R. Raz, A. Shpilka,} Deterministic
  polynomial identity testing in non commutative models, {\sl
    Computational Complexity,}14(1):1-19, 2005.

\bibitem[SV10]{SV10} {\sc A. Shpilka, I. Volkovich,} On the Relation
  between Polynomial Identity Testing and Finding Variable Disjoint
  Factors, {\sl ICALP }, 2010.

\bibitem[SY00]{SY99} {\sc Arto Salomaa and Sheng Yu}, On the
  decomposition of finite languages, \emph{In Proc. Developments in
    Language Theory, Foundations, Applications, and Perspectives,
    editors Grzegorz Rozenberg and Wolfgang Thomas, World
    Scientific,}, 22-31, 2000.

\end{thebibliography}
\end{document}